\documentclass[draft]{IEEEtran}
\onecolumn

\usepackage[noadjust]{cite}
\usepackage[english]{babel}
\usepackage{amsmath,amsthm,mathrsfs}
\usepackage{amsfonts}
\usepackage{mathabx}
\usepackage{flushend}
\usepackage[final,hyperfootnotes=false]{hyperref}
\usepackage[final]{graphicx}
\usepackage{subfigure}
\usepackage{epstopdf}
\usepackage{pgfplots}
\usepackage{tikz}
\usetikzlibrary{arrows,positioning,fit,backgrounds}
\usetikzlibrary{calc}

\usepackage{array}
\newcolumntype{L}[1]{>{\raggedright\let\newline\\\arraybackslash\hspace{0pt}}m{#1}}
\newcolumntype{C}[1]{>{\centering\let\newline\\\arraybackslash\hspace{0pt}}m{#1}}
\newcolumntype{R}[1]{>{\raggedleft\let\newline\\\arraybackslash\hspace{0pt}}m{#1}}

\allowdisplaybreaks

\newtheorem{thm}{Theorem}%[section]
\newtheorem{cor}{Corollary}

\newtheorem{lem}{Lemma}
\newtheorem{prop}{Proposition}

\theoremstyle{definition}

\theoremstyle{remark}
\newtheorem{rem}{Remark}
\DeclareMathOperator*{\rmd}{d}

\newcommand{\var}{{\rm Var}}

\usepackage{pifont}

\newcommand\blfootnote[1]{%
  \begingroup
  \renewcommand\thefootnote{}\footnote{#1}%
  \addtocounter{footnote}{-1}%
  \endgroup
}

\begin{document}
\title{Dispersion Bound for the Wyner-Ahlswede-K\"orner Network
via Reverse Hypercontractivity on Types}

\author{Jingbo Liu
%Dept. of Electrical Eng., Princeton University, NJ 08544\\
}

\maketitle
\begin{abstract}
%\emph{THIS PAPER IS ELIGIBLE FOR THE STUDENT PAPER AWARD}.\footnote{Disclaimer: at the time of submission I have defended (12/12/2017) but not awarded the degree (01/20/2018).}
This paper introduces a new converse machinery for a challenging class of distributed source-type problems
(e.g.\ distributed source coding, common randomness generation, or hypothesis testing with communication constraints), through the example of the Wyner-Ahlswede-K\"orner network.
Using the functional-entropic duality and the reverse hypercontractivity of the transposition semigroup,
we lower bound the error probability for each joint type.
Then by averaging the error probability over types,
we lower bound the $c$-dispersion (which characterizes
the second-order behavior of the weighted sum of the rates of the two compressors when a nonvanishing error probability is small) as the variance of the gradient of $\inf_{P_{U|X}}\{cH(Y|U)+I(U;X)\}$
with respect to $Q_{XY}$, the per-letter side information and source distribution.
In comparison, using standard achievability arguments based on the method of types, we upper-bound the $c$-dispersion as the variance of $c\imath_{Y|U}(Y|U)+\imath_{U;X}(U;X)$,
which improves the existing upper bounds
but has a gap to the aforementioned lower bound.
\end{abstract}

\blfootnote{
This paper was presented in part at the 2018 IEEE International Symposium
on Information Theory (ISIT), July 2018.
Jingbo Liu is now a postdoctoral research associate at the MIT Institute for Data, Systems, and Society (IDSS), Cambridge, MA, 02139 ({jingbo}@mit.edu).
The majority of this work was done while he was at Princeton University, Princeton, NJ, 08544.
}

\IEEEpeerreviewmaketitle
\section{Introduction}
\begin{figure}[h!]
  \centering
\begin{tikzpicture}
[node distance=0.6cm,minimum height=10mm,minimum width=12mm,arw/.style={->,>=stealth'}]
  \node[rectangle,draw,rounded corners] (D) {~~~Decoder~~~};
  \node[rectangle,draw,rounded corners] (E1) [left=1.4cm of D] {Encoder 2};
  \node[rectangle,draw,rounded corners] (E2) [below =0.5cm of D] {Encoder 1};
  \node[rectangle] (Y) [left =0.4cm of E1] {$Y^n$};
  \node[rectangle] (X) [left =0.4cm of E2] {$X^n$};
  \node[rectangle] (Yh) [right =0.4cm of D] {$\hat{Y}^n$};

  \draw [arw] (Y) to node[midway,above]{} (E1);
  \draw [arw] (X) to node[midway,right]{} (E2);
  \draw [arw] (E1) to node[midway,above]{$W_2$} (D);
  \draw [arw] (E2) to node[midway,right]{$W_1$} (D);
  \draw [arw] (D) to node[midway,left]{} (Yh);
\end{tikzpicture}
\caption{Source coding with compressed side information}
\label{pump:f_agk}
\end{figure}
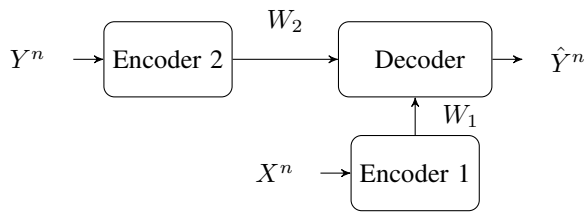
In the Wyner-Ahlswede-K\"orner (WAK) problem \cite{wyner1975side}\cite{ahlswede_bounds_cond1976}, a source $Y^n$ and a side information $X^n$ are compressed separately as integers $W_2$ and $W_1$, respectively,
and a decoder reconstructs $Y^n$ as $\hat{Y}^n$.
Consider the discrete memoryless setting where the per-letter source distribution is $Q_{XY}$,
for any $c>0$, define
\begin{align}
\phi_c(Q_{XY}):=
\inf_{P_{U|X}}\{cH(Y|U)+I(U;X)\}
\label{e_phi}
\end{align}
where $(U,X,Y)\sim P_{U|X}Q_{XY}$.
The following \emph{strong converse} result was proved in \cite{ahlswede_bounds_cond1976} using the \emph{blowing-up lemma}:
if the error probability $\mathbb{P}[\hat{Y}^n\neq Y^n]$ is below some $\epsilon\in (0,1)$, then
\begin{align}
\ln|\mathcal{W}_1|+c\ln|\mathcal{W}_2|\ge
n\phi_c(Q_{XY})-O\left(\sqrt{n}\ln^{3/2}n\right).
\label{e_known}
\end{align}
where the cardinality of the auxiliary can be bounded by $|\mathcal{U}|\le |\mathcal{X}|+2$.
The first-order term in \eqref{e_known} is the precise single-letter characterization \cite{wyner1975side}\cite{ahlswede_bounds_cond1976}.
Note that for any $c<1$, we have $\phi_c(Q_{XY})=cH(Y)$
by the data processing inequality.
Moreover, $\ln|\mathcal{W}_1|+c\ln|\mathcal{W}_2|\ge c\ln |\mathcal{W}_1\times\mathcal{W}_2|\ge cnH(Y)-O(\sqrt{n})$,
which follows simply from a method of type analysis  \cite{csiszar2011information} of the single source compression problem.
Therefore the only nontrivial case is $c\ge 1$.

While recent research has succeeded in studying the \emph{second-order} rates for various single-user and selected multiuser problems (see e.g.\ \cite{polyanskiy2010channel}\cite{tan2014asymptotic}),
it remained a formidable challenge to precisely characterize second-order term in \eqref{e_known}.
Indeed, \cite[Section~9.2.2, 9.2.3]{tan2014asymptotic} listed the second-order rate in WAK as a major open problem, since previous converse techniques (e.g.\ straightforward uses of method of types, information spectrum methods, or meta-converses)
appear insufficient for cases where the auxiliary random variable satisfies a Markov chain.
In fact, the WAK problem represents a typical challenge in a class of distributed source coding problems (or more generally, distributed source-type problems including common randomness generation \cite{ahlswede1998common} or hypothesis testing \cite{ahlswede1986hypothesis}) involving side information (a.k.a.\ a helper).
Recently, Watanabe \cite{watanabe17} examined the converse bound obtained by taking limits in the Gray-Wyner network,
yielding a strong converse for WAK but not appearing to improve the second-order term.
In \cite{tw2018}, Tyagi and Watanabe proposed an approach of dealing with such Markov chain constraints,
by replacing it with a bound on the conditional mutual information and then taking the limits.
While their approach yields strong converses in interesting applications such as Wyner-Ziv and wiretap channels,
it is not clear whether such an approach yields sharper second-order estimates in \eqref{e_known}.

To our knowledge, the first proof of an $O(\sqrt{n})$ second-order converse was due to \cite{ISIT_lhv2017_no_url}.
The $O(\sqrt{n})$ rate is sharp for $\epsilon>1/2$
since an $\Omega(\sqrt{n})$ bound follows by applying the central limit theorem to the standard random coding argument.
In fact, \cite{ISIT_lhv2017_no_url} proposed a new and widely applicable converse technique based on functional-entropic duality and reverse hypercontractivity,
and it appears that the entire blowing-up lemma business for strong converses \cite{ahlswede_bounds_cond1976}\cite{csiszar2011information}
was merely suboptimal approximations of this approach.
After the publication of \cite{ISIT_lhv2017_no_url},
Zhou-Tan-Yu-Motani \cite{zhou} and Oohama \cite{oohama} improved a previous technique of Oohama and claimed that an $O(\sqrt{n})$ converse for WAK also follows from that technique.

The idea of \cite{ISIT_lhv2017_no_url} is roughly described as follows: first we note that an entropic quantity related to $\phi_c$ has an equivalent functional version
\eqref{e_dual} which contains quantities such as $\int\ln f\rmd P$.
If one directly takes $f$ to be the indicator function of a decoding set, then generally $\int\ln f\rmd P=-\infty$ which is useless.
However, using a machinery called reverse hypercontractivity, we design some ``magic operator'' $\Lambda$ such that $\int\ln (\Lambda f)\rmd P\ge \ln \int f\rmd P$, and $\int f\rmd P$ is the probability of correct decoding which we desired.
For all source and channel networks where a strong converse was proved in  \cite{csiszar2011information}, we can now bound the second-order term as $C\sqrt{n\ln\frac{1}{1-\epsilon}}$.

One deficiency of the machinery of \cite{ISIT_lhv2017_no_url} is the inability to deal with the case of $\epsilon<1/2$.
In this paper, we integrate the reverse hypercontractivity machinery with the method of types,
which is capable of showing,
among others,
that for sufficiently small (but independent of $n$) $\epsilon$, the $O\left(\sqrt{n}\ln^{3/2}n\right)$ term in
\eqref{e_known} is $-\Omega(\sqrt{n})$.
More precisely,
under a mild differentiability assumption on $\phi_c(\cdot)$,
we show the following lower bound on the \emph{$c$-dispersion}, defined as the left side of \eqref{e_dis}:
\begin{align}
&\quad\lim_{\epsilon\downarrow 0}\limsup_{n\to\infty}
\frac{[\inf
\{\ln|\mathcal{W}_1|+c\ln|\mathcal{W}_2|\}
-n\phi_c(Q_{XY})]^2}
{2n\ln\frac{1}{\epsilon}}
\nonumber\\
&\ge {\rm Var}\left(\left.\nabla \phi_c\right|_{Q_{XY}}(X,Y)\right)
\label{e_dis1}
\\
&=
{\rm Var}(\mathbb{E}[c\imath_{Y|U}(Y|U)+\imath_{U;X}(U;X)|XY])
\label{e_dis}
\end{align}
where the infimum is over codes for which $\mathbb{P}[\mathcal{E}_n]\le\epsilon$,
$(U,X,Y)\sim Q_{UXY}:=Q_{U|X}Q_{XY}$,
$Q_{U|X}$ is \emph{any} infimizer for \eqref{e_phi},
and we used the notations
\begin{align}
\imath_{Y|U}(y|u)&:=\frac{1}{Q_{Y|U}(y|u)},
\quad\forall (y,u);
\\
\imath_{U;X}(u;x)&:=\frac{Q_{X|U}(x|u)}{Q_X(x)},
\quad\forall (u,x).
\end{align}
We can take $|\mathcal{U}|\le |\mathcal{X}|+2$ \cite{ahlswede_bounds_cond1976}.
We remark that the second-order bound $C\sqrt{n\ln\frac{1}{1-\epsilon}}$ in \cite{ISIT_lhv2017_no_url} does not give a nontrivial bound for the dispersion,
whereas the present bound \eqref{e_dis1} is analogous to the dispersion formula in most other previously solved problems from the network information theory.

Apart from reverse hypercontractivity and functional-entropic duality,
another interesting ingredient in our proof is an argument in analyzing certain information quantity for the equiprobable distribution on a type class (Lemma~\ref{lem_au}).
We perform an algebraic expansion employing the symmetry of the type class, but different from the standard tensorization argument for the i.i.d.\ distribution.
This gives rise to a certain martingale, whose variance equals the gap to the same quantity evaluated for the i.i.d.\ distribution.

On the achievability side, a previously published upper-bound on the $c$-dispersion is
$\left(\sqrt{{\rm Var}(c\imath_{Y|U}(Y|U))}+\sqrt{{\rm Var}(\imath_{U;X}(U;X))}\right)^2$ \cite{watanabe_tan}.
In this paper we use the method of types to show an improved upper bound of
\begin{align}
{\rm Var}\left(
c\imath_{Y|U}(Y|U)+\imath_{U;X}(U;X)\right)
\label{e_ach}
\end{align}
%\begin{align}
%&\quad\mathbb{E}\left[
%{\rm Var}\left(\left.c\imath_{Y|U}(Y|U)
%+\imath_{U;X}(U;X)\right|UX\right)\right]
%\nonumber\\
%&+
%{\rm Var}\left(
%\mathbb{E}\left[
%\left.c\imath_{Y|U}(Y|U)+\imath_{U;X}(U;X)\right|X
%\right]
%\right)
%\label{e_ach}
%\end{align}
which generally has a gap to \eqref{e_dis}.
Our achievability proof uses standard techniques, so the main methodological contribution of the paper is the converse part.
We remark that a dispersion formula of the type \eqref{e_ach} (in lieu of \eqref{e_dis}) has appeared in noisy lossy source coding \cite{kv2013}.
However, the auxiliary $U$ in WAK assumes the role of a time sharing variable (as opposed to a reconstruction alphabet), so it may seem venturesome to draw an analogy and conjecture \eqref{e_ach} to be the true dispersion.

\emph{Notation.}
Given an alphabet $\mathcal{Y}$, define
$\mathcal{H}_+(\mathcal{Y})$ as the set of all nonnegative functions on $\mathcal{Y}$, and $\mathcal{H}_{[0,1]}(\mathcal{Y})$ the set of functions from $\mathcal{Y}$ to $[0,1]$.
For $f\in\mathcal{H}_+(\mathcal{Y})$, define $P(f):=\mathbb{E}_P[f]$,
and define $P_{Y|X}(f):=\mathbb{E}_{P_{Y|X=\cdot}}[f(Y)]$ as a function on $\mathcal{H}_+(\mathcal{X})$.
%(see e.g.\ \eqref{e_dual}).
Given an $n$-type $P_{XY}$,
let $\mathcal{T}_n(P_X)$ be the set of all $x^n$ with type $P_X$, and $\mathcal{T}_{x^n}(P_{Y|X})$ the set of all $y^n$ such that $(x^n,y^n)$ is type $P_{XY}$.
The total variation distance is denoted by $|P-Q|$.
We use $P_X\to P_{Y|X}\to P_Y$ to define an output distribution $P_X$ for a given input and a random transformation.
Define the Gaussian tail probability ${\rm Q}(t):=\int_t^{\infty}\frac{1}{\sqrt{2\pi}}e^{-\frac{x^2}{2}}\rmd x$, for $t\in \mathbb{R}$.
The bases for all exponentials and entropic quantities are natural.
Unless otherwise stated, the constants used in bounding (e.g.\ $E$,$F$,$G$,$\lambda$) may depend on $c$ and $Q_{XY}$.
Given a probability measure $Q$ on $\mathcal{X}$,
and a functional $\phi$ on the probability simplex $\Delta_{\mathcal{X}}$,
the gradient $\nabla\left.\phi\right|_Q$ is a function on $\mathcal{X}$,
and $\langle \nabla\left.\phi\right|_Q,\,Q\rangle
:=\int\nabla\left.\phi\right|_Q\rmd Q$.

\section{Main Results}\label{sec_main}
\subsection{Converse}
All converse analysis in this paper assumes finite alphabets $\mathcal{X}$ and $\mathcal{Y}$, and that $\phi_c(\cdot)$ (defined in \eqref{e_phi}) has bounded second derivatives in a neighborhood of $Q_{XY}$.
The main ingredient of the converse proof is the following bound in the case where the source sequences are equiprobably distributed on a given type class.
\begin{lem}\label{thm_fc}
Given $Q_{XY}$ and $c\ge1$,
there exists $\lambda\in (0,1)$ and $E>0$ %(depending on $c$ and $Q_{XY}$)
such that the following holds:
for any $n$-type $P_{XY}$ such that $|P_{XY}-Q_{XY}|\le\lambda$,
let $(X^n,Y^n)$ be equiprobable on the type class $\mathcal{T}_n(P_{XY})$.
If there exists a WAK coding scheme with error probability $\epsilon\in (0,1)$,
then
\begin{align}
&\quad\ln|\mathcal{W}_1|+c\ln|\mathcal{W}_2|
\nonumber\\
&\ge
n\phi_c(P_{XY})-2c\sqrt{\frac{n}
{\min_xP_X(x)}\ln\frac{1}{1-\epsilon}}
\nonumber\\
&\quad-E\ln n+c\ln(1-\epsilon).
\label{e_main}
\end{align}
\end{lem}
By averaging the error probability bounded in Lemma~\ref{thm_fc} over the types,
we obtain the following converse for stationary memoryless sources:
\begin{thm}\label{cor_dis}
Fix $c\ge1$, $D\in\mathbb{R}$,
and $Q_{XY}$.
Let $(X^n,Y^n)\sim Q_{XY}^{\otimes n}$.
If a WAK coding scheme satisfies
\begin{align}
\ln|\mathcal{W}_1|+c\ln|\mathcal{W}_2|
&\le n\phi_c(Q_{XY})+D\sqrt{n}
\label{e_rate}
\end{align}
for all $n$ then we lower bound the error probability
\begin{align}
\liminf_{n\to\infty}\mathbb{P}[\mathcal{E}_n]
\ge \sup_{\delta\in (0,1)}\delta{\rm Q}\left(\tfrac{D+\sqrt{\frac{2}{\min_x Q_X(x)}\ln\frac{1}{1-\delta}}}
{\sqrt{{\rm Var}(\left.\nabla\phi_c\right|_{Q_{XY}}(X,Y))}}
\right).
\label{e_disp_bound}
\end{align}
In particular, the $c$-dispersion define as the left side of \eqref{e_dis} is lower bounded by ${\rm Var}(\left.\nabla\phi_c\right|_{Q_{XY}}(X,Y))$.
\end{thm}
Proofs of Lemma~\ref{thm_fc} and \eqref{e_disp_bound} are given in Section~\ref{sec_proof}.
After \eqref{e_disp_bound} is established,
using the fact that $\lim_{r\to\infty}\frac{\ln{\rm Q}(r)}{r^2}=-\frac{1}{2}$
we can lower-bound the $c$-dispersion as 
\begin{align}
\sup_{\delta\in (0,1)}\lim_{D\to\infty}
\limsup_{n\to\infty}
\frac{D^2n}{n\cdot\frac{\left(D+
\sqrt{\frac{2}{\min_xQ_X(x)}\ln\frac{1}{1-\delta}}\right)^2}
{{\rm Var}(\left.\nabla\phi_c\right|_{Q_{XY}}(X,Y))}}
={\rm Var}(\left.\nabla\phi_c\right|_{Q_{XY}}(X,Y))
\end{align}
establishing the second claim in Theorem~\ref{cor_dis}.

\subsection{Achievability}
As alluded, we also use the method of types (in lieu of the information spectrum approach of \cite{watanabe_tan}) to obtain the following improved upper bound on $c$-dispersion.
\begin{thm}\label{thm_achieve}
Fix $Q_{XY}$ on finite alphabets, $c\ge 1$, and $D\in \mathbb{R}$.
There exists a WAK scheme scheme for each $n$ such that
\begin{align}
\ln|\mathcal{W}_1|+c\ln|\mathcal{W}_2|
&\le n\phi_c(Q_{XY})+D\sqrt{n};
\\
\limsup_{n\to\infty}\mathbb{P}[\mathcal{E}_n]
&\le
Q\left(\frac{D}{\sqrt{V}}\right),
\end{align}
where we defined $V$ as \eqref{e_ach}.
\end{thm}
Proof of Theorem~\ref{thm_achieve} is given in Section~\ref{sec_achieve}.
We remark that generally there is a gap between the bounds on $c$-dispersion in Theorem~\ref{cor_dis} and Theorem~\ref{thm_achieve}.

\section{Basic Properties of the Single-Letter Expression}
To better interpret our results and prepare for the proofs,
it is instructive to understand some of the basic properties of the single-letter rate expressions.
To fix ideas,
let us first recall the situation in the simpler and well-studied problem of lossy compression of a single source (see e.g.\ \cite{kostina2012fixed}).
In that problem, we are given a single source with per-letter distribution $Q_X$,
and a per-letter distortion $\rmd\colon \mathcal{U}\times \mathcal{X}\to \mathbb{R}$ on the reconstruction alphabet and the source alphabet.
If $P_{U|X}$ is an optimizer for $\varphi_{\lambda}(Q_X):=\inf_{P_{U|X}}\{I(U;X)+\lambda \mathbb{E}[\rmd(U;X)]\}$,
then the stationarity condition implies that
\begin{align}
\imath_{U;X}(u;x)+\lambda\rmd(u,x)
\end{align}
is 
\begin{enumerate}
\item independent of $u$, $P_{U|X}Q_X$-a.s.;
\item equal to $\nabla \varphi_{\lambda}|_{Q_X}(x)$,
regardless of the choice of the optimal $P_{U|X}$.
It is known (e.g.\ \cite{kostina2012fixed}) that the dispersion equals ${\rm Var}\left(\nabla \varphi_{\lambda}|_{Q_X}(X)\right)$.
\end{enumerate}

Now in WAK, our lower and upper bounds on the dispersion in Theorem~\ref{cor_dis} and Theorem~\ref{thm_achieve}, although different,
are both analogous to the solution in single-user lossy source coding in certain senses.
More precisely, 
we observe Proposition~\ref{prop_ind} and Proposition~\ref{prop_grad} below, 
which are parallel to the two properties listed above for single-user lossy compression.
\begin{prop}\label{prop_ind}
For any $Q_{XY}$ on finite alphabets and $c\ge 1$,
let $P_{U|X}$ be optimal in the definition of $\phi_c(\cdot)$ in \eqref{e_phi}
and suppose that $\mathcal{U}$ is finite\footnote{This is merely a simplifying assumption and is without loss of generality. 
Indeed, Carath\'eodory's theorem implies that one can take $|\mathcal{U}|\le|\mathcal{X}|+2$ \cite{ahlswede_bounds_cond1976}.}.
Then $\mathbb{E}[c\imath_{Y|U}(Y|U)+\imath_{U;X}(U;X)|UX]$
is independent of $U$ almost surely.
\end{prop}
\begin{proof}
Let us introduce the notations
\begin{align}
f(P'_{U|X})&=cH(Y'|U')+I(U';X'),
\\
g(u,x,y)&=c\imath_{Y|U}(y|u)+\imath_{U;X}(u;x),
\end{align}
where $P'_{U|X}$ is any random transformation from $\mathcal{X}$ to $\mathcal{U}$, 
and $(U',X',Y')\sim P'_{U|X}Q_{XY}$.
The first order term in the Taylor expansion of $f(P'_{U|X})-f(P_{U|X})$ equals 
$\sum_{u,x,y}g(u,x,y)(P'_{U|X}-P_{U|X})(u,x)Q_{XY}(x,y)$,
which must vanish by the optimality of $P_{U|X}$.
In particular, fix any $x$ such that $Q_X(x)\neq 0$,
and suppose that $(P'_{U|X}-P_{U|X})(u|x')=0$ for all $u$, unless $x'=x$.
We then have
\begin{align}
\sum_u\mathbb{E}[g(U,X,Y)|U=u,X=x]
(P'_{U|X}-P_{U|X})(u,x)=0.
\end{align}
This shows that $\left(\sum_u\mathbb{E}[g(U,X,Y)|U=u,X=x]\right)_{u\in\mathcal{U}}$
is orthogonal to the subspace of vectors whose coordinates sum to zero,
so itself must be a vector with constant coordinates (depending possibly on $x$ but not $u$).
\end{proof}
We remark that $c\imath_{Y|U}(y|u)+\imath_{U;X}(u;x)$ is generally not independent of $u$. Below is an explicit example.
\\
{\bf Binary symmetric sources:} Suppose that $X$ and $Y$ are both equiprobable on $\{-1,1\}$ and $\mathbb{E}[XY]=\rho$.
Consider any $c\in[\rho^{-2},\infty)$.
Remark that $c<\rho^{-2}$ is the degenerate case since $\rho^2$ is the strong data processing constant.
Let $U$ be equiprobable on $\{-1,1\}$ and such that $U-X-Y$ and $\mathbb{E}[UX]=\eta$,
where $\eta$ is defined as the solution to
\begin{align}
c=\frac{\ln\frac{1+\eta}{1-\eta}}
{\rho\ln\frac{1+\eta\rho}{1-\eta\rho}}.
\label{e_c}
\end{align}
Note that the $\eta$ satisfying \eqref{e_c} maximizes $cH(Y|U)+I(U;X)$ for the given $c$ and $\rho$.
Using Mrs.\ Gerbers lemma (see e.g.\ \cite{el2011network}) one can show that such $P_{U|X}$ is an infimizer for \eqref{e_phi}.
We can compute that
\begin{align}
c\imath_{Y|U}(1|1)+\imath_{U;X}(1;1)
&=
c\ln\frac{2}{1+\rho\eta}+\ln(1+\eta);
\\
c\imath_{Y|U}(1|-1)+\imath_{U;X}(-1;1)
&=
c\ln\frac{2}{1-\rho\eta}+\ln(1-\eta).
\end{align}
%hence
%\begin{align}
%\mathbb{E}[c\imath_{Y|U}(1|U)+\imath_{U;X}(;1)|X=Y=1]
%=\frac{c}{2}\ln\frac{4}{1-\rho^2\eta^2}
%+\frac{\eta c}{2}\ln\frac{1-\rho\eta}{1+\rho\eta}
%+\frac{1}{2}\ln(1-\eta^2)+\frac{\eta}{2}\ln\frac{1+\eta}{1-\eta}.
%\end{align}

\begin{prop}\label{prop_grad}
For any $Q_{XY}$ on finite alphabets and $c\ge 1$,
suppose that $\phi_c(\cdot)$ is differentiable at $Q_{XY}$.
Then for any optimal $P_{U|X}$,
  \begin{align}
  \mathbb{E}[c\imath_{Y|U}(Y|U)+\imath_{U;X}(U;X)|X,Y]
  =\nabla \left.\phi_c\right|_Q(X,Y).
  \label{e14}
  \end{align}
  In particular, the left side does not depend on the choice of the optimal $P_{U|X}$.
\end{prop}
\begin{proof}
Recall that $\phi_c(Q_{XY})$ is defined as the infimum of $cH(Y|U)+I(U;X)$ over $P_{U|X}$.
As a general fact, 
the derivative of an infimum equals the partial derivative of the objective function evaluated at an optimizer, 
under suitable differentiability assumption (see e.g.\ the calculation in \cite[Lemma 13]{lccv2015}.
Now for fixed $P_{U|X}$, the partial derivative of $cH(Y|U)+I(U;X)$ with respect to $Q_{XY}$ equals the left side of \eqref{e14}.
\end{proof}

\section{Proof of the Converse}\label{sec_proof}
%\subsection{Preliminaries}
%{\bf add a preliminary on taking the derivative.}
\subsection{Overview}
\subsubsection{Fixed Composition Argument} 
Imagine that a genie tells all encoders and decoders the joint type of the source, and they all design coding strategies for each type.
This oracle setup gives a converse to the original problem in the stationary memoryless setting.
For the purpose of second-order rate analysis,
the intuition behind many previously solved problems
(including the challenging ones such as the Gray-Wyner network \cite{watanabe2017}) 
may be described as follows.
One roughly sees a dichotomy: for some ``good types'', the error probability essentially equals 1, and for the rest ``bad types'', the error probability is essentially 0.
Thus the total error probability is tightly approximated by the probability of those ``bad types''.

\subsubsection{Challenge of Fixed Composition Argument for WAK}
In network information theory problems where the auxiliary in the single-letter expression satisfies a Markov chain,
a straightforward fixed composition argument as described above does not seem to give even a strong converse (unsurprisingly, since otherwise the authors of \cite{csiszar2011information}\cite{ahlswede_bounds_cond1976}
who are familiar with the method of types would not need the blowing-up lemma for strong converses in \cite{csiszar2011information}).
Moreover, when the blowing-up lemma is applied to a type class,
the second-order term is still $O(\sqrt{n}\ln^{3/2}n)$, no better than the i.i.d.\ case.
As a side remark,
some computations by the author indicate that the above 0-1 dichotomy may not be true when the auxiliary satisfies a Markov chain.
%which is probably the core challenge in closing the gap between the dispersion bounds in Section~\ref{sec_main}.

\subsubsection{New Machineries}
In the present paper, we perform fixed composition analysis in the nonvanishing error regime,
and the second-order rate is improved to
$O(\sqrt{n})$.
In lieu of the blowing-up lemma,
we use the dual representation of $\phi_c(\cdot)$ as well as reverse hypercontractivity
--both ingredients integrate naturally and are responsible for the improved rates.
We remark that the $O(\sqrt{n})$ rate is the same order as the i.i.d.\ case in \cite{ISIT_lhv2017_no_url}.
Because it is not $o(\sqrt{n})$, we do not get a clean bound on the second-order term for \emph{each} $\epsilon\in(0,1)$;
we have a bound involving nuisance constants depending on $Q_X$ (see \eqref{e_disp_bound}).
However, the beauty of the \emph{dispersion} \eqref{e_dis} is that the nuisance constant disappears as $\epsilon\to0$.
The technical part of the paper is to show that there exists the ``magic operator'' $\Lambda$ we used in the proof of the converse.
This is done in Section~\ref{app}, where we use
the estimate of the modified log-Sobolev constant in \cite{gao2003}.

\subsection{Proof of Lemma~\ref{thm_fc}}
Suppose that $f\colon \mathcal{X}^n\to \mathcal{W}_1$,
$g\colon \mathcal{Y}^n\to \mathcal{W}_2$ are the encoders, and $V\colon \mathcal{W}_1\times \mathcal{W}_2\mapsto \mathcal{Y}^n$ denotes the decoder.
For each $w\in\mathcal{W}_1$,
define the ``correctly decodable set'':
\begin{align}
\mathcal{B}_w:=\{y^n\colon V(w,g(y^n))=y^n\}.
\end{align}
Let $P_{X^nY^n}$ be the equiprobable distribution on $\mathcal{T}_n(P_{XY})$.
By the assumption,
\begin{align}
\int P_{Y^n|X^n}[\mathcal{B}_{f(x^n)}|x^n]\rmd P_{X^n}(x^n)
\ge 1-\epsilon.
\label{e_assump}
\end{align}
Next, we lower bound the error probability using the functional inequality and reverse hypercontractivity approach.
We introduce a ``magic'' linear operator $\Lambda_{n,t}\colon \mathcal{H}_+(\mathcal{Y})\to\mathcal{H}_+(\mathcal{Y})$,
apply it to the indicator function of a decodable set,
and plug the resulting function into the functional inequality.
To streamline the presentation, we postpone the definition of $\Lambda_{n,t}$ to \eqref{e_lambda}.
The key properties of $\Lambda_{n,t}$,
the proofs of which deferred to Section~\ref{app},
are the following:
for $f\in \mathcal{H}_{[0,1]}(\mathcal{Y}^n)$ and $t=1/\sqrt{n}$,
\begin{itemize}
\item {\bf Lower bound} \eqref{e_et} $P_{Y^n|X^n}(\ln\Lambda_{n,t}f)\ge O(\sqrt{n})\ln P_{Y^n|X^n}(f)$,
\item {\bf Upper bound} \eqref{e_ub}
$P_{Y^n}(\Lambda_{n,t}f)\le \exp(O(\sqrt{n}))P_{Y^n}(f)$.
\end{itemize}
Now, for any $t>0$,
\begin{align}
&\quad(1-\epsilon)^{c\left(1+\frac{1}{t}\right)}
\nonumber\\
&\le \int P_{Y^n|X^n}^{c\left(1+\frac{1}{t}\right)}
[\mathcal{B}_{f(x^n)}|x^n]\rmd P_{X^n}(x^n)
\label{e_jensen}
\\
&= \sum_{w\in\mathcal{W}_1}\int_{x^n\colon f(x^n)=w}
P_{Y^n|X^n}^{c\left(1+\frac{1}{t}\right)}
[\mathcal{B}_w|x^n]\rmd P_{X^n}(x^n)
\\
&\le |\mathcal{W}_1|\int
P_{Y^n|X^n}^{c\left(1+\frac{1}{t}\right)}
[\mathcal{B}_{w^*}|x^n]\rmd P_{X^n}(x^n)
\label{e_different}
\\
&\le  |\mathcal{W}_1|\int
\exp\left(cP_{Y^n|X^n=x^n}(\ln \Lambda_{n,t}1_{\mathcal{B}_{w^*}})\right)
\rmd P_{X^n}
\label{e_use_rhc}
\\
&\le e^d|\mathcal{W}_1|P_{Y^n}^c(\Lambda_{n,t}
1_{\mathcal{B}_{w^*}})
\label{e_duality}
\\
&\le e^d |\mathcal{W}_1|\exp_e\left(\frac{nt}{\min_xP_X(x)}
\right)
P_{Y^n}^c[\mathcal{B}_{w^*}]
\label{e_use_ub}
\\
&\le e^d
|\mathcal{W}_1|\exp_e\left(\frac{nt}{\min_xP_X(x)}
\right)|\mathcal{W}_2|^c\cdot|\mathcal{T}_n(P_Y)|^{-c}.
\label{e_w2}
\end{align}
Here,
\begin{itemize}
  \item \eqref{e_jensen} used Jensen's inequality.
  \item For \eqref{e_different}, we can clearly choose some $w^*\in\mathcal{W}_1$ such that this line holds.
  \item \eqref{e_use_rhc} used the precise form of the lower bound stated above.
        This is the reverse hypercontractivity step.
  \item For \eqref{e_duality}, we defined\footnote{It is interesting to note that the largest $c>1$ for which $d=0$ equals the reciprocal of the strong data processing constant.}
  \begin{align}
  d:=\sup_{S_{X^n}}\{cD(S_{Y^n}\|P_{Y^n})
  -D(S_{X^n}\|P_{X^n})\}
  \end{align}
  where $S_{X^n}\to P_{Y^n|X^n}\to S_{Y^n}$.
  A basic functional-entropic duality result (see e.g.\ \cite{ISIT_lccv_2016}) is that
  \begin{align}
  d=\sup_{f\in\mathcal{H}_+(\mathcal{Y}^n)}
  \left\{
  \ln P_{X^n}(e^{cP_{Y^n|X^n}(\ln f)})
  -c\ln P_{Y^n}(f)
  \right\}
  \label{e_dual}
  \end{align}
  which is the key functional-entropic duality step.
  \item \eqref{e_use_ub} used the precise form of the upper bound stated above.
  \item \eqref{e_w2} used $|\mathcal{B}_{w^*}|\le|\mathcal{W}_2|$.
\end{itemize}
We thus obtain
\begin{align}
&\quad \ln|\mathcal{W}_1|+c\ln|\mathcal{W}_2|
\nonumber\\
&\ge -d+c\ln|\mathcal{T}_n(P_Y)|
\nonumber\\
&\quad-\inf_{t>0}
\left\{\frac{nt}{\min_xP_X(x)}+
c\left(1+\frac{1}{t}\right)\ln\frac{1}{1-\epsilon}
\right\}
\\
&\ge -d-c\ln|\mathcal{T}_n(P_Y)|+c\ln(1-\epsilon)
\nonumber\\
&\quad-2c\sqrt{\frac{n}{\min_xP_X(x)}
\ln\frac{1}{1-\epsilon}}.
\label{e28}
\end{align}
Lemma~\ref{lem_au} bounds $-d-c\ln|\mathcal{T}_n(P_Y)|$, and we are done.
\begin{rem}
From the proof we see that the result continues to hold if the $\mathcal{Y}$-encoder is allowed to access the message of the $\mathcal{X}$-encoder: $g\colon \mathcal{Y}\times \mathcal{W}_1\to \mathcal{W}_2$.
\end{rem}
\begin{rem}
We used Jensen inequality to get \eqref{e_different} from \eqref{e_assump}.
In contrast, \cite{ahlswede_bounds_cond1976} used a reverse Markov inequality, essentially deducing from \eqref{e_assump} that
\begin{align}
P_{X^n}[x^n\colon P_{Y^n|X^n=x^n}
[\mathcal{B}_{w^*}]\ge 1-\epsilon']
\ge \frac{\epsilon'-\epsilon}{\epsilon'|\mathcal{W}_1|}
\label{e_theirs}
\end{align}
which gives rise to a new parameter $\epsilon'$ to be optimized.
It is possible to follow \eqref{e_theirs} with the functional approach as we did in \cite{ISIT_lhv2017_no_url}.
However, proceeding with \eqref{e_different} is more natural and better manifests the simplicity and flexibility of the functional approach \cite{ISIT_lhv2017_no_url}.
\end{rem}

\subsection{Proof Theorem~\ref{cor_dis}}
Let $P_{XY}$ be an arbitrary $n$-type such that $|P_{XY}-Q_{XY}|\le \lambda$ as in Lemma~\ref{thm_fc}.
Then if the error probability conditioned on type $P_{XY}$ is less than $\delta\in (0,1/2)$, we have
\begin{align}
n\phi_c(Q_{XY})+D\sqrt{n}
&\ge
n\phi_c(P_{XY})-2c\sqrt{\tfrac{n}{\min_xP_X(x)}
\ln\tfrac{1}{1-\delta}}
\nonumber\\
&\quad-E\ln n-c\ln2.
\end{align}
Remark that the last two terms will be immaterial for the asymptotic analysis.
Note that by the Taylor expansion, there exists $F>0$ (depending on $Q_{XY}$ and $c$) such that for any $P_{XY}$ in the $\lambda$-neighborhood of $Q_{XY}$,
\begin{align}
\phi_c(P_{XY})&\ge\phi_c(Q_{XY})
+\langle \nabla \left.\phi_c\right|_Q, P_{XY}-Q_{XY}\rangle
\nonumber\\
&\quad -F|P_{XY}-Q_{XY}|^2.
\end{align}
Combining the two bounds above,
the error conditioned on type $P_{XY}$ exceeds $\delta$ if
\begin{align}
&\quad n\langle \nabla \left.\phi_c\right|_Q,
P_{XY}-Q_{XY}\rangle
\nonumber\\
&>
D\sqrt{n}+2c\sqrt{\frac{n}{\min_x P_X(x)}\ln\frac{1}{1-\delta}}
\nonumber\\
&\quad+E\ln n+nF|P_{XY}-Q_{XY}|^2+c\ln2.
\label{e_14}
\end{align}
Now particularize $P_{XY}$ to be the empirical distribution of $(X^n,Y^n)\sim Q_{XY}^{\otimes n}$.
Then with probability $1-O(e^{-n^{1/3}})$ we have $|P_{XY}-Q_{XY}|<n^{-1/3}$ (by Hoeffding's inequality) and $\frac{1}{\min_xP_X(x)}<\frac{2}{\min_x Q_X(x)}$,
and \eqref{e_14} holds if (for some $G>0$)
\begin{align}
&\quad\sum_{i=1}^n
\left.\phi_c\right|_Q(X_i,Y_i)
-\mathbb{E}[\left.\phi_c\right|_Q(X,Y)]
\nonumber\\
&>
D\sqrt{n}
+2c\sqrt{\frac{2n}{\min_x Q_X(x)}\ln\frac{1}{1-\delta}}
+Gn^{1/3}.
\end{align}
%Note that by the optimality of $U$ in \eqref{e_phi} implies (see e.g.\ a similar argument in \cite{kostina2012fixed})
%\begin{align}
%\nabla \left.\phi_c\right|_{Q}(x,y)=c\imath_{Y|U}(y|u)
%+\imath_{U;X}(u;x)
%\end{align}
%for any $x,y$, where the right side above is in fact independent of $u$.
Thus by CLT, we conclude that the probability of type with error exceeding $\delta$ is at least
\begin{align}
{\rm Q}\left(
\frac{D+\sqrt{\frac{2}{\min_xP_X(x)}\ln\frac{1}{1-\delta}}}
{\sqrt{{\rm Var}(\nabla \left.\phi_c\right|_{Q}(X,Y))}}
\right)-o(1).
\end{align}
The $c$-dispersion is then lower bounded by
\begin{align}
&\quad\lim_{D\to\infty}\frac{D^2}{2\ln\frac{1}{\delta}
-2\ln {\rm Q}\left(
\frac{D+\sqrt{\frac{2}{\min_xP_X(x)}\ln\frac{1}{1-\delta}}}
{\sqrt{{\rm Var}(\nabla \left.\phi_c\right|_{Q}(X,Y))}}
\right)}
\nonumber\\
&={\rm Var}(\nabla \left.\phi_c\right|_{Q}(X,Y)).
\end{align}

\subsection{Single-letterization on Types}
Given an $n$-type $P_{XY}$,
let $P_{X^nY^n}$ be the equiprobable distribution on $\mathcal{T}_n (P_{XY})$,
and let $P_{Y^n|X^n}$ be the induced random transformation defined for distributions supported on $\mathcal{T}_n(P_X)$.
Let
\begin{align}
\psi_{c,n}(P_{XY}):=
\inf_{S_{X^n}}\{cH(S_{Y^n})+D(S_{X^n}\|P_{X^n})\}.
\label{e_conj}
\end{align}
Here, the infimum is over $S_{X^n}$ supported on $\mathcal{T}_n(P_X)$, and we have set $S_{Y^n}$ by
$S_{X^n}\to P_{Y^n|X^n}\to S_{Y^n}$.
\begin{lem}\label{lem_au}
Given $Q_{XY}$ and $c\ge 1$, there exists $\lambda\in(0,1)$ and $E>0$ such that for any $n$-type $P_{XY}\colon |P_{XY}-Q_{XY}|<\lambda$,
\begin{align}
\psi_{c,n}(P_{XY})
\ge
n\phi_c(P_{XY})-E\ln n.
\label{e_au}
\end{align}
\end{lem}
\begin{proof}
Under the assumption that $\phi_c$ has bounded second derivatives in a neighborhood of $Q_{XY}$,
there exists $\lambda\in (0,1)$ and $E'>0$ large enough such that
\begin{align}
\phi_c(S_{XY})
&\ge
\phi_c(P_{XY})
+\langle\nabla\left.\phi_c\right|_{P_{XY}},S_{XY}-P_{XY}\rangle
\nonumber\\
&\quad-E'\|S_{XY}-P_{XY}\|^2
\label{e_ta}
\end{align}
for any $P_{XY}\colon |P_{XY}-Q_{XY}|\le\lambda$ and any $S_{XY}$ in the probability simplex (the Taylor expansion proves \eqref{e_ta} for $S_{XY}$ in a neighborhood of $P_{XY}$. Then using the boundedness of $\phi_c(\cdot)$ we can extend \eqref{e_ta} to all $S_{XY}$ by choosing $E'$ large enough).
Here $\|\cdot\|$ denotes the $\ell_2$ norm, although any norm admitting an inner product would work.
Consider any $S_{X^n}$ supported on $\mathcal{T}_n(P_X)$,
and put $S_{X^nY^n}=S_{X^n}P_{Y^n|X^n}$.
Let $I$ be equiprobable on $\{1,\dots,n\}$ and independent of $(X^n,Y^n)$ under $S$.
Let $X_{\setminus I}$ denote the coordinates excluding the $I$-th one.
We have
\begin{align}
H(S_{Y^n})&=H(S_{Y_I|I}|S_I)
+H(S_{Y_{\setminus I}|IY_I}|S_{IY_I})
\\
&\ge H(S_{Y_I|I}|S_I)
+H(S_{Y_{\setminus I}|IX_IY_I}|S_{IX_IY_I}),
\label{e_h}
\end{align}
\begin{align}
&\quad D(S_{X^n}\|P_{X^n})
\nonumber\\
&=
D(S_{X_I|I}\|P_{X_I}|S_I)
+D(S_{X_{\setminus I}|IX_I}\|P_{X_{\setminus I}|X_I}|S_{IX_I})
\\
&=D(S_{X_I|I}\|P_{X_I}|S_I)
+D(S_{X_{\setminus I}|IX_IY_I}\|
P_{X_{\setminus I}|X_IY_I}|S_{IX_IY_I})
\label{e_markov}
\end{align}
where \eqref{e_markov} follows from $X_{\setminus I}-IX_I-Y_I$ under $S$.
Noting that $P_{X_I}=P_X$
and $S_{IX_IY_I}=S_{I|X_I}P_{XY}$, we can bound the weighted sum of the first terms in the above two expansions:
\begin{align}
cH(S_{Y_I|I}|S_I)+D(S_{X_I|I}\|P_{X_I}|S_I)
\ge
\phi_c(P_{XY}).
\end{align}
To bound the weighted sum of the second terms in \eqref{e_h} and \eqref{e_markov},
for any $(x,y)$, define
\begin{align}
P_{XY}^{xy}(x',y')
:=\tfrac{1}{n-1}[nP_{XY}(x',y')-1_{(x',y')=(x,y)}],
\quad \forall (x',y').
\end{align}
That is, $P_{XY}^{xy}$ denotes the $(n-1)$-type obtained by removing one pair $(x,y)$ from  sequences of the type $n$-type $P_{XY}$.
Then
\begin{align}
&\quad cH(S_{Y_{\setminus I}|IX_IY_I}|S_{IX_IY_I})
+D(S_{X_{\setminus I}|IX_IY_I}\|
P_{X_{\setminus I}|X_IY_I}|S_{IX_IY_I})
\nonumber\\
&\ge \sum_{x,y,i}\psi_{c,n-1}(P_{XY}^{xy})
S_{IX_IY_I}(i,x,y)
\\
&= \sum_{x,y,i}\psi_{c,n-1}(P_{XY}^{xy})
P_{XY}(x,y).
\end{align}
Summarizing and iterating,
\begin{align}
\psi_{c,n}(P_{XY})
&\ge \phi_c(P_{XY})
+\mathbb{E}[\psi_{c,n-1}(P_{XY}+\Delta_1)]
\\
&\ge \dots
\\
&\ge \sum_{k=0}^{n-1}\mathbb{E}[\phi_c(P_{XY}
+\Delta_1+\dots+\Delta_k)]
\end{align}
where we defined the sequence $\Delta_1,\Delta_2,\dots$ of random vectors in the following way:
conditioned on $\Delta_1,\Delta_k$, denote $S_{XY}:=P_{XY}+\sum_{i=1}^k\Delta_k$,
and then $\Delta_{k+1}:=S_{XY}^{xy}-S_{XY}$ with probability $S_{XY}(x,y)$ for each $(x,y)$.
Using \eqref{e_ta}, and noting that $\Delta_1+\dots+\Delta_k$ is a zero mean martingale, we have
\begin{align}
\psi_{c,n}(P_{XY})
&\ge
n\phi_c(P_{XY})-
E'\sum_{k=1}^{n-1}
\mathbb{E}\|\Delta_1+\dots+\Delta_k\|^2
\nonumber\\
&\ge n\phi_c(P_{XY})-
E'\sum_{k=1}^{n-1}(n-k)\mathbb{E}\|\Delta_k\|^2
\\
&=n\phi_c(P_{XY})-
E'\sum_{k=1}^{n-1}(n-k)\cdot\frac{4}{(n-k)^2}
\label{e_p1}
\\
&=n\phi_c(P_{XY})-
4E'(1+\ln(n-1))
\end{align}
where \eqref{e_p1} follows from the fact that $\|\Delta_k\|\le|\Delta_k|\le\frac{2}{n-k}$ with probability 1.
Taking $E=10E'$ completes the proof.
\end{proof}

\section{Proof of the Achievability}\label{sec_achieve}
%We first prove for the case where an oracle tells two encoders the compression lengths for the two encoders needed based on the type $\widehat{P}_{X^n}$;
%after that, we will show how to tweak the scheme so that the compression lengths are independent of $\widehat{P}_{X^n}$.

We first make a few preliminary observations about the optimization problem in the definition of $\phi_c$:
\begin{prop}\label{prop_prelim}
Let $c>0$, and $Q_{XY}$ be fully supported on the finite set $\mathcal{X}\times\mathcal{Y}$ (that is, $P_{XY}(x,y)>0$ for each $(x,y)$).
Let $P_{U|X}^{\star}$ be an infimizer for \eqref{e_phi}.
Assume without loss of generality that $P_U^{\star}$ is fully supported on some finite set $\mathcal{U}$.
Then
\begin{enumerate}
  \item\label{p1} $P_{U|X=x}^{\star}$ is also fuly supported on $\mathcal{U}$ for each $x$.
  \item\label{p3} as long as $I(U;X)>0$,
  $(U,X)\sim P_{U|X}^{\star}Q_X$, we have
  \begin{align}
  \left.\nabla_{P_{U|X}}I(U;X)\right|_{P_{U|X}^{\star}}\neq 0.
  \label{e56}
  \end{align}
\end{enumerate}
\end{prop}
\begin{proof}
The proofs follow from the first-order optimality condition.
For the first claim,
note that if $P_U$ is fully supported and $P_{U|X}(u|x)=0$,
we have
$
\left.\frac{\partial}{\partial P_{U|X}(u|x)}I(U;X)\right|_{P_{U|X}^{\star}}
=-\infty
$
whereas $
\left.\frac{\partial}{\partial P_{U|X}(u|x)}H(Y|U)\right|_{P_{U|X}^{\star}}
$.
This contradicts the optimality of $P_{U|X}^{\star}$.
For the second claim,
observe that $I(U;X)>0$ under $P_{U|X}^{\star}Q_X$ implies the existence of some $x$ such that $\left(\imath_{U;X}(u;x)\right)_{u\in\mathcal{U}}$,
which equals to $\left.\nabla_{P_{U|X=x}}I(U;X)\right|_{P_{U|X}^{\star}}$
up to an additive constant,
cannot be a vector with constant coordinates.
\end{proof}

We next define the encoders and the decoder for each type of $X^n$,
and perform the error analysis.
One challenge in an obvious strategy is that the compression lengths for the two encoders will have to depend on the type $\widehat{P}_{X^n}$, even though their weighted sum does not vary with the type.
A remedy is to perturb the ``encoder'' $P_{U|X}$ according to $\widehat{P}_{X^n}$, so that each individual compression length does not vary with $\widehat{P}_{X^n}$ either.
We will see that \eqref{e56} guarantees that we can find such a perturbation to ``trade off'' the two compression lengths;
the first-order optimality condition
\begin{align}
\left.\nabla_{P_{U|X}}I(U;X)\right|_{P_{U|X}^{\star}}
+c\cdot\left.\nabla_{P_{U|X}}H(Y|U)\right|_{P_{U|X}^{\star}}=0,
\label{e55}
\end{align}
where $(U,X,Y)\sim P_{U|X}Q_{XY}$,
ensures that such ``trade-off'' does not affect the weighted sum.
\\
{\bf Encoder 1:} Define $M_1$ by
\begin{align}
\log M_1=nI(U^{\star};X^{\star})+n^{0.02}
\end{align}
where
\begin{align}
(U^{\star},X^{\star},Y^{\star})\sim P_{U|X}^{\star}Q_{XY}.
\end{align}
Encoder~1 constructs a codebook consisting of $M_1$ codewords i.i.d.\ according to $P_{U}^{\star\otimes n}$.
Upon observing $X^n$, Encoder~1 will send the index of $\widehat{P}_{X^n}$ (using $O(\log n)$ bits) and then send the index of the codeword.
The codeword is selected in the following way:
For each $\widehat{P}_{X^n}$ satisfying
\begin{align}
|\widehat{P}_{X^n}-Q_X|\le n^{-0.49},
\label{e_dev}
\end{align}
we define below a $P_{U|X}$ which is a perturbation of $P_{U|X}^{\star}$.
Then upon observing $X^n$, the encoder selects the first codeword $U^n$ such that $(U^n,X^n)$ has type $P_{U|X}\widehat{P}_{X^n}$ (if any).

To define such a $P_{U|X}$ associated with each $\widehat{P}_{X^n}$,
we can first pick a fixed $P_{U|X}'$ such that
\begin{align}
\left.\frac{\partial}{\partial t}I(Q_X,\,tP_{U|X}'+(1-t)P_{U|X}^{\star})\right|_{t=0}\neq 0.
\end{align}
This is possible in view of Proposition~\ref{prop_prelim}.
Then take
\begin{align}
P_{U|X}=tP_{U|X}'+(1-t)P_{U|X}^{\star},
\label{e_p_ux}
\end{align}
rounded such that $P_{U|X}\widehat{P}_{X^n}$ is $n$-type,
where
\begin{align}
t:=-\frac{\left<\imath_{U^{\star};X^{\star}},\,
P_{U|X}^{\star}(\widehat{P}_{X^n}-Q_X)\right>}{\left.\frac{\partial}{\partial t}I(Q_X,\,tP_{U|X}'+(1-t)P_{U|X}^{\star})\right|_{t=0}}.
\end{align}
%We will show in the error analysis that \eqref{e_pux}-\eqref{e_pux3} are fulfilled.

{\bf Encoder 2: }Each $y^n$ sequence is mapped randomly to one of $M_2$ bins, where we defined $M_2$ by
\begin{align}
\log M_2=nH(Y^{\star}|U^{\star})+\frac{D}{c}\sqrt{n}+n^{0.03},
\end{align}
and the bin index is sent.

{\bf Decoding rule: }The decoder receives $\widehat{P}_{X^n}$,
and hence knows the previously agreed $P_{U|X}$ as long as \eqref{e_dev} holds.
The decoder selects the $y^n$ sequence in the $m_2$-th bin that minimizes the empirical conditional entropy
$H(\widehat{P}_{y^n|u^n}|\widehat{P}_{u^n})$.

{\bf Error analysis: }
\begin{itemize}
\item Let $\mathcal{E}_0$ be the event that
\begin{align}
|\widehat{P}_{U^nX^nY^n}-P_{U|X}^{\star}Q_{XY}|>n^{-0.49}.
\end{align}
    Using the Bernstein inequality, we can show that with high probability, $|\widehat{P}_{X^nY^n}-Q_{XY}|\le O(n^{-0.49})$, which in turn implies that $\|P_{U|X}-P_{U|X}^{\star}\|\le O(n^{-0.49})$
    (of course, the choice of the norm is immaterial).
    More precisely, we have
\begin{align}
\mathbb{P}[\mathcal{E}_0]=O(e^{-n^{0.001}}).
\end{align}

\item Let $\mathcal{E}_1$ be the event that no codeword is selected by Encoder 1.
    We first argue that $\mathcal{E}_0^c$ guarantees that
\begin{align}
nD(P_{U|X}\|Q_U^{\star}|\widehat{P}_{X^n})+n^{0.01}&\le \log M_1.
\label{e_pux3}
\end{align}
Indeed, $\mathcal{E}_0^c$ implies
\begin{align}
|\widehat{P}_{X^n}-Q_X|\le n^{-0.49},\label{e_pxn}
\end{align}
which in turn implies 
\begin{align}
t=O(n^{-0.49}),
\end{align}
and
\begin{align}
\sup_{x}|P_{U|X=x}-P_{U|X=x}^{\star}|&=O(n^{-0.49}).\label{e_pux2}
\end{align}
Then by the Taylor expansion we have
\begin{align}
&\quad D(P_{U|X}\|P_U^{\star}|\widehat{P}_{X^n})
\nonumber\\
&=I(U^{\star};X^{\star})
+\left<\imath_{U^{\star};X^{\star}},\,
(P_{U|X}-P_{U|X}^{\star})Q_X\right>
\nonumber\\
&\quad+\left<\imath_{U^{\star};X^{\star}},\,
P_{U|X}^{\star}(\widehat{P}_{X^n}-Q_X)\right>
+O(n^{-0.98}).
\label{e_expansion}
\end{align}
The choice of $P_{U|X}$ in \eqref{e_p_ux} ensures that the second and the third terms in \eqref{e_expansion} cancel,
so that \eqref{e_pux3} is fulfilled.

Next, recalled that each codeword is generated according to $P_U^{\star\otimes n}$.
Conditioned on any $\widehat{P}_{X^n}$, a codeword and $X^n$ has the joint type $P_{U|X}\widehat{P}_{X^n}$ with probability $e^{-nD(P_{U|X}\|P_U^{\star}|\widehat{P}_{X^n})-O(\log n)}$.
Hence for any $\widehat{P}_{X^n}$ satisfying \eqref{e_pxn},
using \eqref{e_pux3} we have
\begin{align}
\mathbb{P}[\mathcal{E}_1|\widehat{P}_{X^n}]
&\le
[1-e^{-nD(P_{U|X}\|P_U^{\star}|\widehat{P}_{X^n})-O(\log n)}]^{M_1}
\\
&\le
O(\exp(-e^{0.001n})),
\end{align}
and consequently,
\begin{align}
\mathbb{P}[\mathcal{E}_1\setminus\mathcal{E}_0]
\le
O(\exp(-e^{0.001n})).
\end{align}

\item Let $\mathcal{E}_2$ be the event that there exists some $y'^n\neq Y^n$ such that the conditional entropy for its conditional empirical distribution is smaller:
    \begin{align}
        H(\widehat{P}_{y'^n|U^n}|\widehat{P}_{U^n})
    <H(\widehat{P}_{Y^n|U^n}|\widehat{P}_{U^n}),
    \end{align}
    where $U^n$ denotes the codeword selected by Encoder~1,
    and $y'^n$ and $Y^n$ are assigned to the same bin by Encoder~2.
    Since there are at most $e^{nH(\widehat{P}_{Y^n|U^n}|\widehat{P}_{U^n})}$ sequences in $\mathcal{Y}^n$
    and since each sequence is mapped to the same bin as $Y^n$ with probability $1/M_2$, by the union bound we have
\begin{align}
\mathbb{P}[\mathcal{E}_2\setminus\mathcal{E}_0]
&\le
\mathbb{E}\left[1_{\mathcal{E}_0^c}\wedge
\frac{e^{nH(\widehat{P}_{Y^n|U^n}|\widehat{P}_{U^n})}}{M_2}
\right].
\label{e_exp}
\end{align}
We now simplify \eqref{e_exp} using the Taylor expansion. Under $\mathcal{E}_0^c$, we have
\begin{align}
nH(\widehat{P}_{Y^n|U^n}|\widehat{P}_{U^n})
&=nH(Y^{\star}|U^{\star})
-n\left<\imath_{Y|U},\,\widehat{P}_{U^nY^n}-P_{UY}^{\star}\right>
+O(n^{0.02}).
\\
&=\sum_{i=1}^n \imath_{Y|U}(Y_i|U_i)+O(n^{0.02}).
\end{align}
%where conditioned on $\widehat{P}_{X^n}$,
%\begin{align}
%(U^n,X^n,Y^n)\sim [][][]
%\end{align}
Therefore,
\begin{align}
\mathbb{P}[\mathcal{E}_2\setminus\mathcal{E}_0]
\le  O(e^{-n^{0.04}})+
\mathbb{P}\left[\mathcal{E}_0^c,\,\sum_{i=1}^n \imath_{Y|U}(Y_i|U_i)-nH(Y^{\star}|U^{\star})-\frac{D}{c}\sqrt{n}>0\right]
\label{e_star}
\end{align}

Define a random variable $S:=\sum_{i=1}^n \imath_{Y|U}(Y_i|U_i)-nH(Y^{\star}|U^{\star})$.
Conditioned on each $\widehat{P}_{X^n}$ satisfying \eqref{e_pxn},
notice that
$(U^n,X^n)$ has the empirical distribution $P_{U|X}\widehat{P}_{X^n}$,
and the conditional distribution of $\frac{1}{\sqrt{n}}\left(S-\mathbb{E}[S|\widehat{P}_{X^n}]\right)$ converges to that of a Gaussian distribution.
Under $\mathcal{E}_0^c$, 
the conditional mean is
\begin{align}
\mathbb{E}\left[S\left|\widehat{P}_{X^n}\right.\right]
&=n\left<\imath_{Y|U},\,P_{U|X}\widehat{P}_{X^n}Q_{Y|X}\right>
-nH(Y^{\star}|U^{\star})
\\
&= n\left<\imath_{Y|U}+\frac{1}{c}\imath_{U;X},\,
P_{U|X}\widehat{P}_{X^n}Q_{Y|X}\right>
-nH(Y^{\star}|U^{\star})-\frac{n}{c}I(U^{\star};X^{\star})
\nonumber\\
&\quad -\frac{n}{c}\left<\imath_{U;X},\, P_{U|X}\widehat{P}_{X^n}Q_{Y|X}\right>+\frac{n}{c}I(U^{\star};X^{\star})
\\
&=n\left<\imath_{Y|U}+\frac{1}{c}\imath_{U;X},\,
P_{U|X}\widehat{P}_{X^n}Q_{Y|X}-P_{U|X}^{\star}Q_{XY}\right>
+O(n^{0.02})
\label{e_36}
\\
&=
n\left<\imath_{Y|U}+\frac{1}{c}\imath_{U;X},\,
P_{U|X}^{\star}(\widehat{P}_{X^n}-Q_X)Q_{Y|X}\right>
+O(n^{0.02})
\label{e_37}
\\
&=\sum_{i=1}^n\mathbb{E}_{P^{\star}_{U|X}Q_{Y|X}}
\left[\left.\imath_{Y|U}+\frac{1}{c}\imath_{U;X}\right|X=X_i\right]
-\frac{1}{c}\phi_c(Q_{XY})+O(n^{0.02}),
\label{e38}
\end{align}
where
\begin{itemize}
\item \eqref{e_36} follows from bounding the last two terms in the previous step:
\begin{align}
&\quad\left<\imath_{U;X},\, P_{U|X}\widehat{P}_{X^n}Q_{Y|X}\right>
-I(U^{\star};X^{\star})
\nonumber\\
&=\left<\imath_{U;X},\,(P_{U|X}-P_{U|X}^{\star})Q_{XY}\right>
+\left<\imath_{U;X},\,P_{U|X}^{\star}(\widehat{P}_{X^n}-Q_X)Q_{Y|X}\right>
+O(n^{-0.98})
\\
&=O(n^{-0.98}),
\end{align}
where we used the Taylor expansion and the definition of $P_{U|X}$ in \eqref{e_p_ux}.
\item To see \eqref{e_37},
we note that 
\begin{align}
\left<\imath_{Y|U}+\frac{1}{c}\imath_{U;X},\,
(P_{U|X}-P_{U|X}^{\star})(\widehat{P}_{X^n}-Q_X)Q_{Y|X}
\right>
= O\left(\|P_{U|X}-P_{U|X}^{\star}\|\cdot|\widehat{P}_{X^n}-Q_X|\right)
=O(n^{-0.98})
\end{align}
and moreover, the first order optimality of $P_{U|X}^{\star}$ implies
\begin{align}
\left<\imath_{Y|U}+\frac{1}{c}\imath_{U;X},\,
(P_{U|X}-P_{U|X}^{\star})Q_{XY}\right>
= O(\|P_{U|X}-P_{U|X}^{\star}\|^2)
=O(n^{-0.98}).
\end{align}
These two inequalities show the step from \eqref{e_36} to \eqref{e_37} upon rearrangements.
\end{itemize}

The conditional variance is bounded as
\begin{align}
\frac{1}{n}\var\left(S\left|\widehat{P}_{X^n}\right.\right)
&=\mathbb{E}_{\widehat{P}_{X^n}P_{U|X}}\left[\var_{Q_{Y|X}}
\left(\imath_{Y|U}|UX\right)\right]
\label{e41}
\\
&=\mathbb{E}_{\widehat{P}_{X^n}P_{U|X}}\left[\var_{Q_{Y|X}}
\left(\left.\imath_{Y|U}+\frac{1}{c}\imath_{U;X}\right|UX\right)
\right]
\\
&=\mathbb{E}_{Q_XP^{\star}_{U|X}}\left[\var_{Q_{Y|X}}
\left(\left.\imath_{Y|U}+\frac{1}{c}\imath_{U;X}\right|UX\right)
\right]
+O(n^{-0.49})
\label{e43}
\\
&=\mathbb{E}_{Q_XP^{\star}_{U|X}}\left[\var_{Q_{Y|X}}
\left(\left.\imath_{Y|U}+\frac{1}{c}\imath_{U;X}\right|X\right)
\right]
+O(n^{-0.49})
\label{e44}
\end{align}
where \eqref{e41} follows since the distribution of $S$ depends only on the empirical distribution $(U^n,X^n)$, which is $P_{U|X}\widehat{P}_{X^n}$,
and \eqref{e44} follows from Proposition~\ref{prop_ind}.

Finally, 
for each type $\widehat{P}_{X^n}$ under $\mathcal{E}_0^c$,
by the Berry-Esse\'en central limit theorem
we see that 
\begin{align}
\mathbb{P}\left[\left.\frac{1}{\sqrt{n}}S
-\mathbb{E}\left[\left.\frac{1}{\sqrt{n}}S\right|
\widehat{P}_{X^n}\right]
>\lambda\right|\widehat{P}_{X^n}
\right]
\le \mathbb{P}[G_n>\lambda]+\xi_n,\quad
\forall \lambda\in \mathbb{R}
\end{align}
where $G_n\sim\mathcal{N}(0,\sigma_n^2)$,
with $\sigma_n^2$ defined as the right side of \eqref{e43},
and $\xi_n$ is some $o(1)$ sequence depending only on $Q_{XY}$, $P_{U|X}^{\star}$, and $c$.
Writing in an equivalent way, we have
\begin{align}
\mathbb{P}\left[\left.\frac{1}{\sqrt{n}}S
>\lambda\right|\widehat{P}_{X^n}
\right]
\le \mathbb{P}\left[\left.G_n+
\mathbb{E}\left[\left.\frac{1}{\sqrt{n}}S\right|
\widehat{P}_{X^n}\right]>\lambda\right|\widehat{P}_{X^n}
\right]+\xi_n,\quad
\forall \lambda\in \mathbb{R}.
\end{align}
The result then follows by unconditioning on $\widehat{P}_{X^n}$.
Note that we essentially bounded the variance proxy of $\frac{1}{\sqrt{n}}S$ as
\begin{align}
&\quad\frac{1}{n}\,\mathbb{E}\left[\var\left(S\left|\widehat{P}_{X^n}
\right.\right)\right]
+\frac{1}{n}\,\var\left(\mathbb{E}\left[S\left|\widehat{P}_{X^n}\right.\right]\right)
\nonumber
\\
&=\mathbb{E}_{Q_XP^{\star}_{U|X}}\left[\var_{Q_{Y|X}}
\left(\left.\imath_{Y|U}+\frac{1}{c}\imath_{U;X}\right|X\right)\right]
\nonumber\\
&\quad+\var\left(\mathbb{E}_{P_{U|X}^{\star}Q_{Y|X}}\left[
\left.\imath_{Y|U}+\frac{1}{c}\imath_{U;X}\right|X\right]\right)
+O(n^{-0.49})
\\
&=\var_{Q_{Y|X}}
\left(\imath_{Y|U}+\frac{1}{c}\imath_{U;X}\right)
+O(n^{-0.49}).
\end{align}

%{\bf Define}
%\begin{align}
%\tilde{M}_1(\widehat{P}_{X^n})&=nD(P_{U|X}\|Q_U|\widehat{P}_{X^n})+\delta;
%\\
%\tilde{M}_2(\widehat{P}_{X^nY^nU^n})
%&=nH(\widehat{P}_{Y^n|U^n}|\widehat{P}_{U^n})+\delta
%\end{align}
\end{itemize}

\section{Appendix: Reverse Hypercontractivity for the Transposition Model}\label{app}
In this section we construct the ``magic operator'' $\Lambda_{n,t}$ used in Section~\ref{sec_proof} through several stages.

\subsection{The Transposition Model}
Let $\mathcal{S}=\{1,\dots,n\}$.
Consider a reversible Markov chain where
the state space $\Omega$ consists of the $n!$ permutations of the sequence $(1,2,\dots,n)$,
and the  generator is given by
\begin{align}
L_nf:=
\frac{1}{n}\sum_{1\le i,j\le n}(f\sigma_{ij}-f)
\label{e_ln}
\end{align}
for any real-valued function $f$ on $\Omega$,
where $f\sigma_{ij}$ denotes the composition of two mappings, and $\sigma_{ij}$ denotes the transposition operator.
That is, $\sigma_{ij}$ switches the $i$-th and the $j$-th coordinates of a sequence
for any $s^n\in\Omega$,
\begin{align}
(\sigma_{ij}s^n)_k:=
\left\{
\begin{array}{cc}
  s_i & k=j; \\
  s_j & k=i; \\
  s_k & \textrm{otherwise}.
\end{array}
\right.
\label{e_transposition}
\end{align}
As an alternative interpretation of this Markov chain, whenever a Poisson clock of rate $\frac{1}{n}$ clicks,
an index pair $(i,j)\in \{1,\dots,n\}^2$ is randomly selected and the corresponding coordinates are switched.
Remark that the rate at which each coordinate changes its value roughly equals 1,
which is the same as the semi-simple Markov Chain we used in \cite{ISIT_lhv2017_no_url}.
Functional inequalities such as Poincar\'e, log-Sobolev, and modified log-Sobolev for such a Markov chain have been studied to bound its mixing time under various metrics.
In particular, we recall the following upper bound on the modified log-Sobolev constant in \cite{gao2003}, which was proved using a chain-rule and induction argument:
\begin{thm}[\cite{gao2003}]
Let $P$ be the equiprobable distribution on $\Omega$.
For any $n\ge 2$,
\begin{align}
D(S\|P)\le -\mathbb{E}\left[\left(L_n\log\frac{\rmd S}{\rmd P}\right)(X)\right],
\quad \forall S\ll P,
\end{align}
where $X\sim S$.
\end{thm}
It is known (e.g.\ \cite[Theorem 1.11]{mossel2013reverse}) that a modified log-Sobolev inequality is equivalent to a reverse hypercontractivity of the corresponding Markov semigroup operator
$e^{L_nt}:=\sum_{k=0}^{\infty}\frac{t^k}{k!}L_n^k$.
We thus have
\begin{cor}
In the transposition model,
For any $q<p<1$, $t\ge\ln\frac{1-q}{1-p}$,
and $f\in\mathcal{H}_+(\Omega)$,
\begin{align}
\|e^{L_nt}f\|_{L^q(\Omega)}\ge \|f\|_{L^p(\Omega)}.
\label{e_rhc}
\end{align}
\end{cor}
We remark that the norms in \eqref{e_rhc} are with respect to the equiprobable measure $P$.
By taking the limits, we have
\begin{align}
\|f\|_{L^0(\Omega)}=\exp\left(P(\ln f)\right).
\end{align}

\subsection{Reverse Hypercontractivity on Types}\label{sec_type}
Now consider any finite $\mathcal{Y}$
and a Markov chain with state space $\mathcal{Y}^n$.
With a slight abuse of notation, let $L_n$ also denote the generator of this new Markov chain.
Let $P_Y$ be an $n$-type.
Note that $\mathcal{T}_n(P_Y)$ is invariant under transposition and hence also an invariant subspace for the chain.
We now prove a reverse hypercontractivity for the Markov semigroup operator for this new chain.
Pick any map $\phi\colon \mathcal{S}\to \mathcal{Y}$ such that $|\phi^{-1}(y)|=nP_Y(y)$ for each $y$.
Then the extension $\phi^n$ defines a function $\Omega\to \mathcal{T}_n(P_Y)$.
Now for any $f\in \mathcal{H}_+(\mathcal{Y}^n)$, from \eqref{e_rhc} we have
\begin{align}
\|e^{L_nt}(f\phi^n)\|_{L^q(\Omega)}\ge \|f\phi^n\|_{L^p(\Omega)}.
\label{e_composition}
\end{align}
We claim that \eqref{e_composition} is equivalent to
\begin{align}
\|e^{L_nt}f\|_{L^q(\mathcal{T}_n(P_Y))}\ge \|f\|_{L^p(\mathcal{T}_n(P_Y))}.
\label{e96}
\end{align}
Indeed, $\|f\phi^n\|_{L^p(\Omega)}
=\mathbb{E}^{\frac{1}{p}}\left[(f\phi^n(S^n))^p\right]
=\mathbb{E}^{\frac{1}{p}}[f^p(Y^n)]=\|f\|_{L^p(\mathcal{T}_n(P_Y))}$.
Here, $L^p(\mathcal{T}_n(P_Y))$ is with respect to the equiprobable measure on $\mathcal{T}_n(P_Y)$,
and so the value of $f$ on $\mathcal{Y}^n\setminus \mathcal{T}_n(P_Y)$ is immaterial.
Moreover, from the definitions we can see that $\phi^n$ commutes with transposition,
so $\left(e^{L_nt}(f\phi)\right)(s^n)=(e^{L_nt}f)(\phi^n(s^n))$
for any $s^n\in\Omega$,
and the left sides of \eqref{e_composition} and \eqref{e96} are therefore also equal by the same argument.

We remark that for $P_Y$ not concentrated on a $y\in\mathcal{Y}$ and as $n\to\infty$,
we don't lose too much tightness in the composition step argument,
and the estimate in \eqref{e96} is sharp.
That is, the modified log-Sobolev constant is indeed of the constant order;
the lower bound can be seen by taking linear functions in the corresponding Poincar\'e inequality, which is weaker than the modified log-Sobolev inequality.

\subsection{Conditional Types: the Tensorization Argument}
Let $\mathcal{X}$ and $\mathcal{Y}$ both be finite sets.
For any $x^n\in\mathcal{X}^n$,
define a linear operator $L_{x^n}\colon \mathcal{H}_+(\mathcal{Y}^n)\to \mathcal{H}_+(\mathcal{Y}^n)$ by
\begin{align}
L_{x^n}f:=\sum_{x\in\mathcal{X}}
\frac{1}{n\widehat{P}_{x^n}(x)}\sum_{i,j\colon x_i=x_j=x}(f\sigma_{ij}-f).
\label{e_ln_cond}
\end{align}
where we recall that $\widehat{P}_{x^n}$ denotes the empirical distribution of $x^n$.
Note that
$L_{x^n}$ is the generator of the Markov chain where
independently for each $x\in\mathcal{X}$,
the length $nP_X(x)$ subsequence of $\mathcal{Y}^n$ with indices $\{i\colon x_i=x\}$ is the transposition model in Section~\ref{sec_type}.
Since $L_{x^n}$ is the sum of $|\mathcal{X}|$ generators for transposition models,
the Markov semigroup operator $e^{L_{x^n}t}$
is a tensor product,
which satisfies the reverse hypercontractivity with the same constant,
by the tensorization property (see e.g.\ \cite{mossel2013reverse}).
Therefore for any $n$-type $P_{XY}$, $x^n\in \mathcal{T}_n(P_X)$,
and $f\colon \mathcal{H}_+(\mathcal{Y}^n)\to \mathcal{H}_+(\mathcal{Y}^n)$,
\begin{align}
\|e^{L_{x^n}t}f\|_{L^q(\mathcal{T}_{x^n}(P_{Y|X}))}\ge \|f\|_{L^p(\mathcal{T}_{x^n}(P_{Y|X}))}.
\label{e97}
\end{align}

\subsection{A Dominating Operator}
The operator in \eqref{e97} depends on $x^n$ and hence cannot be used directly in the proof of Lemma~\ref{thm_fc}.
We now find an upper bound which is independent of $x^n$.
Define a linear operator $\tilde{L}_n\colon\mathcal{H}_+(\mathcal{Y}^n)\to
\mathcal{H}_+(\mathcal{Y}^n)$ by
\begin{align}
\tilde{L}_nf(y^n)
=\frac{1}{n\min_x P_X(x)}
\sum_{1\le i,j\le n} f(\sigma_{ij}y^n).
\end{align}
Note that the summation includes the $i=j$ case, where $\sigma_{ij}$ becomes the identity.
From the general formula
$
\frac{\rmd}{\rmd t}(e^{Lt}f)
=Le^{Lt}f
$
we can see a comparison property: since the matrix of $\tilde{L}_n$ entry-wisely dominate $L_{x^n}$, we have
$
e^{\tilde{L}_nt}f\ge e^{L_{x^n}t}f
$
pointwise for any $t\ge0$ and $f\in\mathcal{H}_+(\mathcal{Y}^n)$.
Now consider
\begin{align}
\Lambda_{n,t}:=e^{\tilde{L}_nt},\quad\forall t>0
\label{e_lambda}
\end{align}
which forms an operator semigroup (although not associated with a conditional expectation).
Now $\Lambda_{n,t}$ is the operator we used in the proof of
Lemma~\ref{thm_fc},
and the two key properties we used are:

{\bf Lower bound:}
for any $f\colon \mathcal{Y}^n\to [0,1]$,
\begin{align}
&\quad\exp(P_{Y^n|X^n=x^n}(\ln\Lambda_{n,t}f))
\nonumber\\
&=\|\Lambda_{n,t}f\|_{L^0(\mathcal{T}_{x^n}(P_{Y|X}))}
\nonumber\\
&\ge \|e^{L_{x^n}t}f\|_{L^0(\mathcal{T}_{x^n}(P_{Y|X}))}
\label{e_56}\\
&\ge \|f\|_{L^{1-e^{-t}}(\mathcal{T}_{x^n}(P_{Y|X}))}
\nonumber\\
&\ge P_{Y^n|X^n=x^n}^{\frac{1}{1-e^{-t}}}(f)
\nonumber\\
&\ge P_{Y^n|X^n=x^n}^{1+\frac{1}{t}}(f)
\label{e_et}
\end{align}
where \eqref{e_et} follows from $e^t\ge 1+t$.
%Note that in \eqref{e_56}-\eqref{e_et} the values of $f$ outside $\mathcal{T}_{x^n}(P_{Y|X})$ are immaterial.

{\bf Upper bound (in fact, equality):}
For any $f\colon \mathcal{Y}^n\to [0,\infty)$,
\begin{align}
&\quad\frac{\rmd}{\rmd t}P_{Y^n}(\Lambda_{n,t}f)
\nonumber\\
&=P_{Y^n}(\tilde{L}_n\Lambda_{n,t}f)
\\
&=\tfrac{1}{n\min_xP_X(x)}\sum_{y^n\in \mathcal{T}_n(P_Y)}
P_{Y^n}(y^n)\sum_{i,j}(\Lambda_{n,t}f)(\sigma_{ij}y^n)
\\
&=\tfrac{n}{\min_xP_X(x)}\sum_{z^n\in \mathcal{T}_n(P_Y)}
P_{Y^n}(z^n)(\Lambda_{n,t}f)(z^n)
\label{e_equi}
\\
&=\tfrac{n}{\min_xP_X(x)}\,P_{Y^n}(\Lambda_{n,t}f)
\end{align}
where \eqref{e_equi} used the fact that $P_{Y^n}$ is the equiprobable distribution on $\mathcal{T}_n(P_Y)$. Thus
\begin{align}
P_{Y^n}(\Lambda_{n,t}f)
&=\exp_e\left(\frac{nt}{\min_x P_X(x)}\right) P_{Y^n}(f).
\label{e_ub}
\end{align}

\section{Discussion}
Through the example of the Wyner-Ahlswede-K\"orner (WAK) network, 
we supplied the mathematical ingredients needed for extending this new converse approach to other potential applications.
There are several distributed source type problems which are very similar to the WAK problem.
For example using a tensor product semigroup for the stationary memoryless settings,
\cite{liu_thesis} 
proved an $O(\sqrt{n})$ second-order converse for common random generation with one-way rate limited communications.
It appears straightforward to upgrade to dispersion bounds and obtain similar results as WAK, 
by following the same steps therein but using the techniques of the present paper.

%\section{Discussions}
%what if the Y-encoder also knows $X$, and the X-encoder receives a copy of the Y-message?

\section{Acknowledgement}
The author would like to thank Prof.\ Ramon van Handel 
and Prof.\ Sergio Verd\'u for their generous support and guidance on research along this line.
Helpful discussions with Prof.\ Ramon van Handel,
Prof.\ Shun Watanabe, Prof.\ Victoria Kostina are greatly appreciated.
This work was supported by NSF grants
CCF-1350595,
CCF-1016625,
CCF-0939370,
and DMS-1148711,
by ARO Grants W911NF-15-1-0479
and W911NF-14-1-0094,
and AFOSR FA9550-15-1-0180,
and by the Center for Science of Information.

\bibliographystyle{ieeetr}
\bibliography{ref_maximal}

\end{document}